\documentclass[12pt,twoside]{article}

 \usepackage{float}
\usepackage{graphicx}
\usepackage{epstopdf}
\usepackage{graphicx}
\usepackage{epic}
\usepackage{multirow}
\usepackage{tikz}
\usepackage{xcolor}
\usepackage{threeparttable}
\usetikzlibrary{arrows,shapes,chains}

\usepackage{graphicx}
\usepackage{makecell}
\renewcommand{\paragraph}{\roman{paragraph}}
\usepackage[a4paper]{geometry}
\makeatletter
\renewcommand\title[1]{\gdef\@title{\reset@font\Large\bfseries #1}}
\renewcommand\section{\@startsection {section}{1}{\z@}%
                                   {-3.5ex \@plus -1ex \@minus -.2ex}%
                                   {2.3ex \@plus.2ex}%
                                   {\normalfont\large\bfseries}}
\renewcommand\subsection{\@startsection{subsection}{2}{\z@}%
                                     {-3ex\@plus -1ex \@minus -.2ex}%
                                     {1.5ex \@plus .2ex}%
                                     {\normalfont\normalsize\bfseries}}
\renewcommand\subsubsection{\@startsection{subsubsection}{3}{\z@}%
                                     {-2.5ex\@plus -1ex \@minus -.2ex}%
                                     {1.5ex \@plus .2ex}%
                                     {\normalfont\normalsize\bfseries}}

\def\@runningauthor{}\newcommand{\runningauthor}[1]{\def\runningauthor{#1}}
\def\@runningtitle{}\newcommand{\runningtitle}[1]{\def\runningtitle{#1}}

\renewcommand{\ps@plain}{%
\renewcommand{\@evenhead}{\footnotesize\scshape \hfill\runningauthor\hfill}
\renewcommand{\@oddhead}{\footnotesize\scshape \hfill\runningtitle\hfill}}

\newcommand{\F}{\mathbb{F}}

\g@addto@macro\bfseries{\boldmath}

\makeatother

\usepackage{amsthm,amsmath,amssymb}
\usepackage{cite}
\usepackage{graphicx}

\usepackage[colorlinks=true,citecolor=black,linkcolor=black,urlcolor=blue]{hyperref}

\theoremstyle{plain}
\newtheorem{theorem}{Theorem}[section]

\newtheorem{lem}[theorem]{Lemma}
\newtheorem{cor}[theorem]{Corollary}
\newtheorem{prop}[theorem]{Proposition}

\theoremstyle{definition}
\newtheorem{definition}[theorem]{Definition}
\newtheorem{example}[theorem]{Example}

\theoremstyle{remark}
\newtheorem{remark}[theorem]{Remark}

\runningauthor{}

\date{}

\begin{document}

\title{\Large{Optimal quaternary linear codes with one-dimensional Hermitian hull and the related EAQECCs}}
\author{Shitao Li\thanks{lishitao0216@163.com}, Minjia Shi\thanks{smjwcl.good@163.com}, Huizhou Liu\thanks{18756027866@163.com}
\thanks{Shitao Li and Minjia Shi are with School of Mathematical Sciences, Anhui University, Hefei, China. Huizhou Liu is with State Grid Anhui Electric Power Co., Ltd., Hefei, China.
}}

\maketitle

\begin{abstract}
Linear codes with small hulls over finite fields have been extensively studied due to their practical applications in computational complexity and information protection.
In this paper, we develop a general method to determine the exact value of $D_4^H(n,k,1)$ for $n\leq 12$ or $k\in \{1,2,3,n-1,n-2,n-3\}$, where $D_4^H(n,k,1)$ denotes the largest minimum distance among all quaternary linear $[n,k]$ codes with one-dimensional Hermitian hull. As a consequence, we solve a conjecture proposed by Mankean and Jitman on the largest minimum distance of a quaternary linear code with one-dimensional Hermitian hull.
 As an application, we construct some binary entanglement-assisted quantum error-correcting codes (EAQECCs) from quaternary linear codes with one-dimensional Hermitian hull. Some of these EAQECCs are optimal codes, and some of them are better than previously known ones.
\end{abstract}

\section{Introduction}
Quantum error-correcting codes (QECCs) play a vital role in quantum computing and quantum communication. Calderbank et al. \cite{PRA-1} and Steane \cite{PRL} introduced an effective method (the CSS construction) for constructing QECCs from classical self-orthogonal codes (or dual-containing codes) over finite fields. However, general linear codes are not suitable for the CSS construction. To avoid this problem,
Brun et al. in \cite{Eaqecc} introduced the concept of entanglement-assisted quantum error-correcting codes (EAQECCs), which a generalization of the quantum stabilizer codes. A very powerful tool for constructing EAQECCs is to determine the hull dimension of a linear code, especially Hermitian hull dimension (see \cite{QIP-111,PRA-2,DCC-hull,sok,Quantum-1,Quantum-2,luo,luo-1,kai-1,zhu-1,zhu-2}).
In this paper, we focus on the most interesting case for applications, namely $p=2$.
We characterize optimal quaternary linear codes with one-dimensional Hermitian hull, and construct good binary EAQECCs from those codes.

The hull of a linear code was first introduced in 1990 by Assmus and Key \cite{A-hull-DM} to classify finite projective planes. Recently, some topics related to hull in coding theory have been widely studied due to their practical applications in computational complexity and information protection. It has been shown that the hull determines the complexity of the algorithms for checking permutation equivalence of two linear codes \cite{J-p-group,Sen-p-e-codes} and for computing the automorphism group of a linear code \cite{J-1,Sen-S-auto-group}. Therefore, it is a hot topic to study linear codes with small hulls.
Let $D^E_q(n,k,l)$ denote the largest minimum distance among all $q$-ary $[n,k]$ linear codes with $l$-dimensional Euclidean hull and $D^H_q(n,k,l)$ denote the largest minimum distance among all $q$-ary $[n,k]$ linear codes with $l$-dimensional Hermitian hull.
The exact values of $D^E_2(n,k,0)$ and $D^E_2(n,n-k,0)$ for $k\in \{1,2,3,4,5\}$ were determined in \cite{D-1,q=2-k=2,A-1,H-1,ter-11-19,AHS-BLCD,q=2-k=5}.
The exact values of $D^E_3(n,k,0)$ and $D^E_3(n,n-k,0)$ for $k\in \{2,3,4\}$ were determined in \cite{1q=3-k=2,AHS-BLCD,ter-11-19}. The exact value of $D^H_4(n,k,0)$ for $k\in \{2,3\}$ were determined in \cite{q=4-k=2,IT-1}. For more related work, readers are refereed to \cite{LCD-a,LCD-b,LCD-c,Dcc-40,ccds}.

Recently, Mankean and Jitman \cite{JAMC-1} determined the exact value of $D_q^E(n,2,1)$ for $q=2,3$. The authors \cite{hull-1} studied the exact values of $D_2^E(n,k,1)$ and $D_2^E(n,n-k,1)$ for $k\leq 30$ or $1\leq k\leq 5$.
Mankean and Jitman \cite{AJM-conjecture} also determined the exact value of $D_4^E(n,2,1)$ for $n\equiv~1,2,4~({\rm mod}~5)$ and made a conjecture for $n\equiv~0,3~({\rm mod}~5)$. For more related work, readers are refereed to \cite{hull-4,hull-2,hull-3,sok}.
In this paper, we consider the next step in the above work. Using the classification in \cite{IT-1,DM} and quaternary linear codes with one-dimensional Hermitian hull related to the simplex codes, we characterize the exact values of $D^H_4(n,k,1)$ and $D^H_4(n,n-k,1)$ for $n\leq 12$ or $k\in\{1,2,3\}$. As a consequence, we solve the conjecture proposed by Mankean and Jitman \cite{AJM-conjecture}. As an application, we construct some new binary EAQECCs from quaternary linear codes with one-dimensional Hermitian hull. For example, we obtain some EAQECCs with the parameters $[[8,3,4;3]]_2$, $[[9,3,5;4]]_2$, $[[10,5,4;3]]_2$, $[[11,2,7;7]]_2$, $[[12,2,8;8]]_2$, $[[12,3,7;7]]_2$, $[[12,7,4;3]]_2$, respectively, while the EAQECCs in \cite{codetables,luo-1} have the parameters $[[8,3,4;4]]_2$, $[[9,3,5;5]]_2$, $[[10,5,4;4]]_2$, $[[11,2,6;5]]_2$, $[[12,2,7;9]]_2$, $[[12,3,7;9]]_2$ and $[[12,7,4;4]]_2$. The EAQECCs we constructed have better minimum distance and smaller amount of entanglement than the best known EAQECCs.
Therefore, we can consider that some of the EAQECCs we have constructed are better than those previously known.

The paper is organized as follows. In Section 2, we give some notations and preliminaries.
In Section 3, we study quaternary linear codes with one-dimensional Hermitian hull related to the simplex codes. In Section 4, we characterize the exact values of $D^H_4(n,k,1)$ and $D^H_4(n,n-k,1)$ for $k\in \{1,2,3\}$. In Section 5, we characterize the exact value of $D^H_4(n,k,1)$ for $n\leq 12$. In Section 6, we construct some EAQECCs from quaternary linear codes with one-dimensional Hermitian hull. In Section 7, we conclude the paper.

\section{Preliminaries}
Let $\F_4=\{0,1,\omega,\omega^2\}$ denote the finite field with 4 elements.
For $\alpha \in \F_4$, the conjugation
of $\alpha$ is defined as $\overline{\alpha}=\alpha^2$.
For a matrix $A = (a_{ij})$, the transpose of $A$ is defined as $A^T=(a_{ji})$, and the conjugate matrix of $A$ is defined as $\overline{A} = (\overline{a_{ij}})$.
For any ${\bf u}\in \F_4^n$, the {\em Hamming weight} of ${\bf u}$ is the number of nonzero components of ${\bf u}$.
A quaternary linear $[n,k,d]$ code $C$ is a $k$-dimensional subspace of $\F_4^n$, where $d$ is the minimum nonzero weight of $C$.
A {\em generator matrix} for a quaternary linear $[n, k]$ code $C$ is any
$k\times n$ matrix $G$ whose rows form a basis for $C$.
A {\em parity-check matrix} for a linear code $C$ is a generator matrix for the dual code $C^{\perp_H}$.
 A {\em monomial} matrix $M$ is a square matrix over $\F_4$ with exactly one nonzero entry in each row and column. Two quaternary linear
$[n, k]$ codes $C_1$ and $C_2$ are {\em equivalent} if there exists a $n\times n$
monomial matrix $M$ over $\F_4$ with $C_2=\{xM~|~x \in C_1\}$.
Let $A_{\omega}$ denote the number of codewords of weight $\omega$ of $C$, where $1\leq i\leq n.$ The sequence $(A_{0},A_1,\ldots,A_n)$ is called the {\em weight distribution} of $C$, and $\sum_{i=0}^nA_ix^i$ is called the {\em weight enumerator} of $C$.

Let $C$ be a quaternary $[n, k, d]$ linear code, and let $S$ be a set of $s$
coordinate positions in $C$. We puncture $C$ by deleting all the coordinates in $S$ in each codeword of $C$. The resulting code is still linear and has length $n-s$. We denote the {\em punctured code} by $C^S$. Consider the set $C(S)$ of codewords which are $0$ on $S$; this set is a subcode of $C$. Puncturing $C(S)$ on $S$ gives a quaternary code of length $n-t$ called the {\em shortened code} on $S$ and denoted $C_S$.

The Hermitian dual code $C^{\perp_H}$ of a quaternary linear $[n,k]$ code $C$ is defined as
$$C^{\perp_H}=\{\textbf v\in \F_4^n~|~\langle \textbf u, \textbf v\rangle_H=0, {\rm for\ all}\ \textbf u\in C \},$$
where $\langle \textbf u, \textbf v\rangle_H=\sum_{i=1}^n u_i\overline{v_i}$ for $\textbf u= (u_1,u_2, \ldots, u_n)$ and $\textbf v = (v_1,v_2, \ldots, v_n)\in \F_4^n$.
Then $C^{\perp_H}$ is a quaternary linear $[n,n-k]$ code.
The {\em Hermitian hull} of a quaternary linear code $C$ is defined as
$${\rm Hull}_{\rm H}(C)= C \cap C^{\perp_H}.$$
It is easy to see that ${\rm Hull}_{\rm H}(C)$ is a quaternary linear code.
Suppose that the dimension of ${\rm Hull}_{\rm H}(C)$ is $l$.
\begin{itemize}
  \item If $l=0$, that is to say, $C\cap C^{\perp_H}=\{\textbf 0\}$, then the code $C$ is a {\em Hermitian linear complementary dual} (Hermitian LCD) code.
  \item If $l=k$, that is to say, $C\subseteq C^{\perp_H}$, then the code $C$ is a {\em Hermitian self-orthogonal} (Hermitian SO) code.
\end{itemize}

The Hermitian hull dimension of a linear code $C$ was determined in \cite{DCC-hull} as follows.

\begin{prop}{\rm\cite[Proposition 3.2]{DCC-hull}}
Let $C$ be quaternary linear $[n,k]$ code with generator matrix $G$. Then
${\rm rank}(G\overline{G}^T)$ is independent of $G$ so that
$$\dim({\rm Hull}_{\rm H}(C))=\dim({\rm Hull}_{\rm H}(C^{\perp_H}))=k-{\rm rank}(G\overline{G}^T).$$
\end{prop}

A quaternary code $C$ is said to be {\em even} if the weights of all codewords of $C$ are even. A sufficient and necessary condition for $C$ to be Hermitian SO is given as follows.

\begin{prop}{\rm \cite[Theorem 1]{JCTA-1}}\label{so}
A quaternary linear code $C$ is Hermitian SO if and only if $C$ is even.
\end{prop}

The {\em Griesmer bound} \cite{Huffman} on a quaternary linear $[n,k,d]$
code is given by
$n\geq \sum_{i=0}^{k-1}\left\lceil \frac{d}{4^i}\right\rceil,$
 where $\lceil a \rceil$ is the least integer greater than or equal to a real number $a$. The {\em sphere-packing bound} \cite{Huffman} on a quaternary linear $[n,k,d]$ code is given by
$4^k\leq \frac{4^n}{\sum_{i=0}^{\lfloor \frac{d-1}{2}\rfloor}3^i\left(\begin{array}{c}
                                                                 n \\
                                                                 i
                                                               \end{array}
\right)}.$

\begin{definition}
For positive integers $n, k$, let
\begin{align*}
  D_4(n,k):= & \max\{d~|~\exists{\rm~ a~quaternary~linear}~[n,k,d]~{\rm code}\}, \\
  D^H_4(n,k,1):= & \max\{d~|~\exists{\rm~ a~quaternary~linear}~[n,k,d]~{\rm code~with~} \dim({\rm Hull}_{\rm H}(C)) =1\}.
\end{align*}
\end{definition}

\begin{theorem}{\rm \cite[Theorem 5]{luo}}\label{thm-luo}
Let $C$ be a quaternary linear $[n,k,d]$ code with Hermitian hull $Hull_H(C)$ of dimension $l$ and let $S\subseteq [n]$ be such that $|S|=s$.
Then the following statements hold.
\begin{itemize}
  \item [(1)] $(Hull_H(C))_S\subseteq Hull_H(C^S)$ and $(Hull_H(C))_S\subseteq Hull_H(C_S)$
  \item [(2)] If $S$ is a subset of an information set of $Hull_H(C)$, then
  $$Hull_H(C^S)=Hull_H(C_S)=(Hull_H(C))_S~with~\dim(Hull_H(C^S))=l-s.$$
\end{itemize}
\end{theorem}

\begin{remark}
Theorem \ref{thm-luo} is a very powerful tool to construct quaternary linear codes with one-dimensional Hermitian hull.
\end{remark}

\section{Quaternary linear codes with one-dimensional Hermitian hull related to the simplex codes}
The $k \times \frac{4^k-1}{3}$ matrix $S_k$ is defined as follows by inductive constructions.
\begin{align*}
  S_1 & =(1) \\
  S_k & =\left(\begin{array}{ccccc}
            S_{k-1} & 0 & S_{k-1} & S_{k-1} & S_{k-1} \\
            {\bf 0} & 1 & {\bf 1} & \omega{\bf 1} & \omega^2{\bf 1}
          \end{array}
  \right)~{\rm if}~k\geq 2
\end{align*}

It is well-known that the matrix $S_k$ generates a quaternary $\left[\frac{4^k-1}{3},k,4^{k-1}\right]$ simplex code, which is an one-weight Hermitian SO code for $k\geq 2$.

\begin{lem}\label{lem-1}
Let $C$ be a quaternary linear $[n,k,d]$ code with generator matrix $G$ for $k\geq 2$. Then $C$ has
one-dimensional Hermitian hull if and only if the code $C_0$ with the following generator matrix
$$G_0=[\underbrace{S_k~|~\cdots~|~S_k}_s|~G ]$$
is a quaternary linear $\left[\frac{4^k-1}{3}s+ n,k,4^{k-1}s+ d\right]$ code with one-dimensional Hermitian hull.
\end{lem}
\begin{proof}
It is well-known that $S_k$ generates a quaternary simplex code, which is an one-weight Hermitian SO $\left[\frac{4^k-1}{3},k,4^{k-1}\right]$ Griesmer code. So
$$G_0\overline{G_0}^T=G\overline{G}^T.$$
Therefore, $C$ has one-dimensional Hermitian hull if and only if $C_0$ has one-dimensional Hermitian hull.
Since $C$ has the minimum distance $d$, $C_0$ has the minimum distance at least $d+4^{k-1}s$. Since the simplex code is an one-weight code, there is at least a codeword of weight $d+4^{k-1}s$ in $C_0$. The converse is also true. This completes the proof.
\end{proof}

\begin{lem}\label{lem-rank}
Let $C_0$ be a quaternary linear $\left[N,K,D\right]$ code with one-dimensional Hermitian hull and the generator matrix
$$G_0=[\underbrace{S_k~|~\cdots~|~S_k}_s|~G ].$$
 If $D>4^{k-1}s$, then the matrix $G$ generates a quaternary linear $[n,k,d]$ code $C$ with one-dimensional Hermitian hull.
\end{lem}

\begin{proof}
By Lemma \ref{lem-1}, we just verify that $C$ has $4^k$ codewords, i.e., ${\rm rank}(G)=k$.
Assume that ${\rm rank}(G)<k$. Since $S_k$ generates an one-weight code, we obtain that $D=2^{k-1}s$, which is a contradiction. This completes the proof.
\end{proof}

Let $h_{k,i}$ be the $i$-th column of the matrix $S_k$. Let $G_{k}({\bf m})$ be a $k\times \sum_{i=1}^{\frac{4^k-1}{3}}m_i$ matrix which consists of $m_i$ columns $h_{k,i}$ for each $i$ as follows:
$$G_{k}({\bf m})=(\underbrace{h_{k,1},\ldots,h_{k,1}}_{m_1},\ldots,
\underbrace{h_{k,\frac{4^k-1}{3}},\ldots,h_{k,\frac{4^k-1}{3}}}_{m_{\frac{4^k-1}{3}}}),$$
where ${\bf m}=(m_1,\ldots,m_{\frac{4^k-1}{3}})$ and $m_i$ is a nonnegative integer.
For a quaternary linear $[n,k,d]$ code with $d(C^{\perp_H})\geq 2$, there exists a vector
${\bf m}=(m_1,\ldots,m_{\frac{4^k-1}{3}})$ such that $C$ is equivalent to the code $C_k({\bf m})$ with the generator matrix $G_k({\bf m})$.

Araya, Harada and Saito \cite{IT-1} proposed an interesting and useful result in order to characterize optimal quaternary Hermitian LCD codes. Now we use this result to characterize optimal quaternary linear codes with one-dimensional Hermitian hull.

\begin{lem}{\rm \cite[Lemma 7]{IT-1}}\label{lem-t}
Suppose that $k\geq 3$, ${\bf m}=(m_1,\ldots,m_{\frac{4^k-1}{3}})$ and $\sum_{i=1}^{\frac{4^k-1}{3}}=n$. If the code $C_k({\bf m})$ has the minimum weight at least $d$, then
$$4d-3n\leq m_i\leq n-\frac{4^{k-1}-1}{3\cdot 4^{k-2}}d,$$
for each $i\in \{1,2,\ldots,\frac{4^k-1}{3}\}$.
\end{lem}

\begin{theorem}\label{thm-no}
If there is no quaternary linear $[4r,3,3r]$ code $C$ with one-dimensional Hermitian hull and $d(C^{\perp_H})\geq 2$, then there is no quaternary linear $[21s+4r,3,16s +3r]$ code $C'$ with one-dimensional Hermitian hull and $d(C'^{\perp_H})\geq 2$ for nonnegative integer $s$.
\end{theorem}

\begin{proof}
Suppose that $C'$ is a quaternary linear $[21s+4r,3,16s +3r]$ code with one-dimensional Hermitian hull and $d(C'^{\perp_H})\geq 2$. Hence there is a vector ${\bf m'}=(m'_1,\ldots,m'_{21})$ such that $C'$ is equivalent to $C_3({\bf m'})$. By Lemma \ref{lem-t}, $m'_i\geq s$. Let ${\bf m}=(m'_1-s,\ldots,m'_{21}-s)$.
According to Lemmas \ref{lem-1} and \ref{lem-rank}, the code
 $C_3({\bf m})$ is a quaternary linear $[4r,3,3r]$ code with one-dimensional Hermitian hull and $d(C^{\perp_H})\geq 2$, which is a contradiction. This completes the proof.
\end{proof}

\begin{theorem}{\rm \cite[Theorem 8]{IT-1}}\label{thm-2}
Suppose that $k\geq 3$. If $n\equiv 0~({\rm mod}~\frac{4^k-1}{3})$, then the quaternary linear code with the parameters $[n,k,D_4(n,k)]$ is Hermitian SO.
\end{theorem}

\begin{proof}
The proof is straightforward by the proof of \cite[Theorem 8]{IT-1}.
\end{proof}

\begin{lem}\label{lem-D-2}
If there is a quaternary linear $[n,k,d]$ code with one-dimensional Hermitian hull for $n>\frac{4^k-1}{3}+1$ and $d>2$, then there is a quaternary linear $[n-2,k,d^*\geq d-2]$ code with one-dimensional Hermitian hull.
\end{lem}

\begin{proof}
Let $C$ be a quaternary linear $[n,k,d]$ code with one-dimensional Hermitian hull and the generator matrix $G$. Since $n>\frac{4^k-1}{3}+1$, the matrix $G$ must have the two columns of the scalar multiple. Without loss of generality, we assume that
$$G=\left[G'~|~{\bf v}^T~|~\alpha {\bf v}^T\right],~{\rm where}~\alpha\in \F^*_4,~{\bf v}\in\F_4^k.$$
It can be checked that $G'\overline{G'}^T=G\overline{G}^T$.
This implies that $${\rm rank}\left(G'\overline{G'}^T\right)={\rm rank}\left(G\overline{G}^T\right)=k-1.$$
Hence the matrix $G'$ generates a quaternary linear $[n-2,k,d^*\geq d-2]$ code with one-dimensional Hermitian hull.
\end{proof}

\section{Optimal quaternary linear $[n,k]$ codes with one-dimensional Hermitian hull}

In this section, we characterize the optimal quaternary linear $[n,k]$ codes with one-dimensional Hermitian hull for $k\in \{1,2,3,n-1,n-2,n-3\}$.

\subsection{Optimal quaternary linear $[n,1]$ and $[n,n-1]$ codes with one-dimensional Hermitian hull}

\begin{theorem}\label{th-k=1}
Suppose that $n\geq 2$. Then
\begin{center}
$
D_4^H(n,1,1)=\left\{
\begin{array}{ll}
n-1,&{\rm if}~n~{\rm is~odd},\\
n,&{\rm if}~n~{\rm is~even}.
\end{array}\right.~
D_4^H(n,n-1,1)=\left\{
\begin{array}{ll}
1,&{\rm if}~n~{\rm is~odd},\\
2,&{\rm if}~n~{\rm is~even}.
\end{array}\right.
$
\end{center}
\end{theorem}

\begin{proof}
By the Griesmer bound, we have $D_4^H(n,1,1)\leq n$ and $D_4^H(n,n-1,1)\leq 2$.
Assume that $n$ is odd. If $C_0$ is a quaternary linear $[n,1,n]$ code, then it follows from Theorem \ref{so} that $C_0$ does not have one-dimensional Hermitian hull. The code $C$ generated by $[0 1 1 \ldots 1]$ is a linear $[n,1,n-1]$ code with one-dimensional Hermitian hull. So $D_4^H(n,1,1)=n-1$.
The dual code $C^{\perp_H}$ of $C$ is a quaternary $[n,n-1,1]$ linear code with one-dimensional Hermitian hull.
If there exists a quaternary linear $[n,n-1,2]$ MDS code with one-dimensional Hermitian hull, then its Hermitian dual code is a quaternary linear $[n,1,n]$ MDS code with one-dimensional Hermitian hull. This contradicts that $D_4^H(n,1,1)=n-1$. Thus $D_4^H(n,n-1,1)=1$.

Assume that $n$ is even. The repetition $[n,1,n]$ code and its dual code are linear codes with one-dimensional Hermitian hull. Hence $D_4^H(n,1,1)=n$ and $D_4^H(n,n-1,1)=2$.
\end{proof}

\subsection{Optimal quaternary linear $[n,2]$ codes with one-dimensional Hermitian hull}

Mankean and Jitman \cite{AJM-conjecture} proposed a conjecture on the minimum distance of quaternary linear codes with one-dimensional Hermitian hull.

\begin{theorem}{\rm \cite[Conjecture 5.2]{AJM-conjecture}}\label{conjecture}
Let $n\geq 3$ be an integer. Then $D_4^H(n,2,1)=\left\lfloor \frac{4n}{5} \right\rfloor-1$ for all positive integers $n\equiv 0,3~({\rm mod}~5)$.
\end{theorem}

\begin{lem}
Suppose that $n\equiv 0,3~({\rm mod}~5)$. If there is a quaternary linear $\left[n,2,\left\lfloor \frac{4n}{5} \right\rfloor\right]$ code with one-dimensional Hermitian hull, then $d(C^{\perp_H})\geq 2$.
\end{lem}
\begin{proof}
Suppose that $n\equiv 0,3~({\rm mod}~5)$. If there is a quaternary linear $\left[n,2,\left\lfloor \frac{4n}{5} \right\rfloor\right]$ code with one-dimensional Hermitian hull, then there is a quaternary linear $\left[n-1,2,\left\lfloor \frac{4n}{5} \right\rfloor\right]$ code with one-dimensional Hermitian hull. This contradicts the Griesmer bound.
\end{proof}

Next, we solve the conjecture by modifying the method proposed by Araya, Harada and Saito \cite{IT-1}.

{\em Proof of Theorem \ref{conjecture}.} Suppose that $n\equiv 0,3~({\rm mod}~5)$. Let $C$ be a quaternary linear $\left[n,2,\left\lfloor \frac{4n}{5} \right\rfloor\right]$ code with $d(C^{\perp_H})\geq 2$. Then there exists ${\bf m}=(m_1,m_2,m_3,m_4,m_5)$ such that the code $C$ is equivalent to the code $C_{k}({\bf m})$ with the following generator matrix

$$G_k({\bf m})=\left( \begin{array}{ccccc}
                        {\bf1}_{m_1} & {\bf0}_{m_2} & {\bf1}_{m_3} & {\bf1}_{m_4} & {\bf1}_{m_5} \\
                        {\bf0}_{m_1} & {\bf1}_{m_2} & {\bf1}_{m_3} & \omega{\bf1}_{m_4} & \omega^2{\bf1}_{m_5}
                      \end{array}
 \right).$$
Without loss of generality, we assume that $m_1\geq 1$ and $m_2\geq 1$.
Then we have
$$G_k({\bf m})\overline{G_k({\bf m})}^T=\left(
\begin{array}{cc}
  m_1+m_3+m_4+m_5 & m_3+\omega^2m_4+\omega m_5 \\
  m_3+\omega m_4+\omega^2 m_5 & m_2+m_3+m_4+m_5
\end{array}
\right).$$
The weight enumerator of the code $C$ is
\begin{align*}
  W_C(x)=&1+ 3x^{m_1+m_3+m_4+m_5}+3x^{m_2+m_3+m_4+m_5}+3x^{m_1+m_2+m_4+m_5}\\
         &+3x^{m_1+m_2+m_3+m_5}+3x^{m_1+m_2+m_3+m_4}.
\end{align*}
Since the code $C$ has the length $n$ and the minimum distance $\left\lfloor \frac{4n}{5} \right\rfloor$,
\begin{center}
$\sum_{i=1}^5m_i  =n$~~and~~$\sum_{i\in \{1,2,3,4,5\}\backslash \{j\}}m_i  \geq \left\lfloor \frac{4n}{5} \right\rfloor,~j=1,2,3,4,5.$
\end{center}
Hence we have
\begin{align}\label{eq-1}
  m_i & \leq n- \left\lfloor \frac{4n}{5} \right\rfloor,~i=1,2,3,4,5.
\end{align}

\begin{itemize}
  \item Suppose that $n\equiv 0~({\rm mod}~5)$, i.e., $n=5s$ for some positive integer $s$.
  From (\ref{eq-1}), we have
  $m_i\leq s,~i=1,2,3,4,5.$
  Then we have
  $n=\sum_{i=1}^5m_i\leq 5s=n.$
  This implies that
  $m_1=m_2=m_3=m_4=m_5=s.$
  Hence $$G_k({\bf m})\overline{G_k({\bf m})}^T=\left(
\begin{array}{cc}
  0 & 0 \\
  0 & 0
\end{array}
\right).$$
It turns out that $C$ is Hermitian SO, which implies that
$$D^H_4(5s,2,1)\leq \left\lfloor \frac{4\times 5s}{5} \right\rfloor-1=4(s-1)+3.$$
On the other hand, there is a quaternary linear $[5,2,3]$ code with one-dimensional Hermitian hull. By Lemma \ref{lem-1}, there is a quaternary linear $[5(s-1)+5,2,4(s-1)+3]$ code with one-dimensional Hermitian hull for $s\geq 1$. Hence
$$D^H_4(5s,2,1)\geq 4(s-1)+3=\left\lfloor \frac{4\times 5s}{5} \right\rfloor-1.$$
Therefore, $D^H_4(5s,2,1)= 4(s-1)+3=\left\lfloor \frac{4\times 5s}{5} \right\rfloor-1.$
  \item Suppose that $n\equiv 3~({\rm mod}~5)$, i.e., $n=5s+3$ for some positive integer $s$. From (\ref{eq-1}), we have
  $m_i\leq s+1,~i=1,2,3,4,5.$
  Then we have
  $$n=\sum_{i=1}^5m_i\leq 5s+5=n+2.$$
  Hence we have
  \begin{align*}
    s-1\leq m_i&\leq s+1,~i=1,2,3,4,5,\\
    |\{i\in\{1,2,3,4,5\}~|~m_i=s+1\}| & = |\{i\in\{1,2,3,4,5\}~|~m_i=s-1\}|+3,\\
    |\{i\in\{1,2,3,4,5\}~|~m_i=s\}| &=2-2|\{i\in\{1,2,3,4,5\}~|~m_i=s-1\}|.
  \end{align*}
  Without loss of generality, let us assume that $m_1\leq m_2$. This yields that there are 11 possibilities for
${\bf m}=(m_1,m_2,m_3,m_4,m_5)$, where the results are listed in Table 1.
The codes in Table 1 are either Hermitian LCD or Hermitian SO.
This implies that
$$D^H_4(5s,2,1)\leq \left\lfloor \frac{4(5s+3)}{5} \right\rfloor-1=4s+1.$$
On the other hand, there is a quaternary linear $[3,2,1]$ code with one-dimensional Hermitian hull. By Lemma \ref{lem-1}, there is a quaternary linear $[5s+3,2,4s+1]$ code with one-dimensional Hermitian hull for $s\geq 0$. Hence
$$D^H_4(5s,2,1)\geq 4s+1=\left\lfloor \frac{4(5s+3)}{5} \right\rfloor-1.$$
Therefore, $D^H_4(5s,2,1)= 4s+1=\left\lfloor \frac{4(5s+3)}{5} \right\rfloor-1.$
\end{itemize}
This completes the proof.
$\hfill\qedsymbol$

\begin{center}
{\small
\begin{tabular}{c|c|c}
\multicolumn{3}{c}{{\rm Table 1: $n=5s+3$}}\\
\hline
   ${\bf m}=(m_1,m_2,m_3,m_4,m_5)$  &  $G_k({\bf m})\overline{G_k({\bf m})}^T$ & $\dim({\rm Hull_H}(C_k({\bf m})))$\\
    \hline\hline
 $(s-1,s+1,s+1,s+1,s+1)$ &  $\left(\begin{array}{cc}
  0 & 0 \\
  0 & 0
  \end{array}
\right)$ & 2\\

$(s+1,s+1,s-1,s+1,s+1)$ &  $\left(\begin{array}{cc}
  0 & s-1 \\
  s-1 & 0
  \end{array}
\right)$& 0~{\rm or}~2\\
$(s+1,s+1,s+1,s-1,s+1)$ &  $\left(\begin{array}{cc}
  0 & 0 \\
  0 & 0
  \end{array}
\right)$& 2\\
$(s+1,s+1,s+1,s+1,s-1)$ &  $\left(\begin{array}{cc}
  0 & 0 \\
  0 & 0
  \end{array}
\right)$& 2\\
$(s,s,s+1,s+1,s+1)$ &  $\left(\begin{array}{cc}
  1 & 0 \\
  0 & 1
  \end{array}
\right)$& 0\\

$(s+1,s,s,s+1,s+1)$ &  $\left(\begin{array}{cc}
  1 & 1 \\
  1 & 0
  \end{array}
\right)$& 0\\
$(s+1,s,s+1,s,s+1)$ &  $\left(\begin{array}{cc}
  1 & \omega^2 \\
  \omega & 0
  \end{array}
\right)$& 0\\
$(s+1,s,s+1,s+1,s)$ &  $\left(\begin{array}{cc}
  1 & \omega \\
  \omega^2 & 0
  \end{array}
\right)$& 0\\
$(s+1,s+1,s,s,s+1)$ &  $\left(\begin{array}{cc}
  0 & \omega \\
  \omega^2 & 0
  \end{array}
\right)$& 0\\
$(s+1,s+1,s,s+1,s)$ &  $\left(\begin{array}{cc}
  0 & \omega^2 \\
  \omega & 0
  \end{array}
\right)$& 0\\
$(s+1,s+1,s+1,s,s)$ &  $\left(\begin{array}{cc}
  0 & 1 \\
  1 & 0
  \end{array}
\right)$& 0\\
    \hline
\end{tabular}}
\end{center}

\subsection{Optimal quaternary linear $[n,3]$ codes with one-dimensional Hermitian hull}
The aim of this section is to establish the following theorem,
which is one of the main results in this paper.
\begin{theorem}\label{thm-1}
Suppose that $n\geq 4$. Then
$$D_4^H(n,3,1)=\left\{\begin{array}{ll}
                          \left\lfloor \frac{16n}{21}\right\rfloor, & if~n~\equiv~1,9,13,14,17,18,19~({\rm mod}~21), \\
                           &  \\
                          \left\lfloor \frac{16n}{21}\right\rfloor-1, & if~n~\equiv~0,2,3,4,6,7,8,10,11,12,15,16,20~({\rm mod}~21).
                        \end{array}
\right.$$
\end{theorem}

By the Griesmer bound, we have
\begin{center}
\begin{tabular}{cc|cc|cc}
\multicolumn{6}{c}{{\rm Table 2: The Griesmer bound on $D(n,3),n\geq 4$}}\\
\hline
   $n$ & $D(n,3)$& $n$ & $D(n,3)$&$n$ & $D(n,3)$\\
    \hline\hline
   $21s$ &$16s$& $21s+7$ &$16s+4$ &$21s+14$ &$16s+10$ \\
   $21s+1$ &$16s$& $21s+8$ &$16s+5$&$21s+15$ &$16s+11$\\
   $21s+2$ &$16s$& $21s+9$ &$16s+6$&$21s+16$ &$16s+12$ \\
 $21s+3$ &$16s+1$& $21s+10$ &$16s+7$&$21s+17$ &$16s+12$ \\
 $21s+4$ &$16s+2$& $21s+11$ &$16s+8$&$21s+18$ &$16s+13$ \\
 $21s+5$ &$16s+3$& $21s+12$ &$16s+8$&$21s+19$ &$16s+14$ \\
 $21s+6$ &$16s+4$& $21s+13$ &$16s+9$&$21s+20$ &$16s+15$ \\
    \hline
\end{tabular}
\end{center}

Bouyukliev, Grassl and Varbanov \cite{DM} completed the classification of optimal quaternary $[n,3]$ codes for $n\leq 35$. Araya, Harada and Saito \cite{IT-1} partially completed the classification of optimal quaternary $[n,3]$ codes for $n\geq 36$.
Based on the classification, we find
$D_4^H(n,3,1)=D_4(n,3)$ for $$n=4,7,8,9,10,11,12,13,14,17,18,19,22,23,24,25,28,29,30,33,34,35.$$
We also find $D_4^H(n,3,1)<D_4(n,3)$ for
$$n=5,6,15,16,20,21,26,27,31,32,36,48,52.$$

We list some generator matrices of quaternary linear $[n,3,D^H_4(n,3,1)]$ codes with one-dimensional Hermitian hull for $n\leq 35$.

\setlength{\arraycolsep}{0.9pt}

\begin{prop}\label{prop-1}
There are quaternary linear $[4,3,2]$, $[7,3,4]$, $[8,3,5]$, $[9,3,6]$, $[12,3,8]$, $[13,3,9]$, $[14,3,10]$, $[17,3,12]$, $[18,3,13]$, $[19,3,14]$, $[22,3,16]$, $[23,3,16]$, $[24,3,17]$ codes with one-dimensional Hermitian hull, which have the following generator matrix, respectively.
$$
\begin{array}{l}
 G_{[4,3,2]}=\left( \begin{array}{cccc}
                          1  & 0  &   0  &   1\\
    0  &   1   &  0   &  1\\
    0   &  0 &    1  &   1
                     \end{array}
 \right),
 G_{[7,3,4]}=\left( \begin{array}{ccccccc}
       1 &0  & 0  & 1 &  \omega  &1&1\\
  0  & 1  & 0  & \omega^2&  \omega^2 & 0   &\omega\\
  0 &  0 &  1 & \omega  & 0  & \omega^2&   1
                     \end{array}
 \right),
 G_{[8,3,5]}=\left( \begin{array}{ccccccccc}
       1&0&0&0&\omega&\omega^2&\omega&1\\
0&1&0&1&\omega&1&0&1\\
0&0&1&1&1&1&\omega&\omega
                     \end{array}
 \right),\\
 G_{[9,3,6]}=\left( \begin{array}{cccccccccc}
       1&0&0&1&1&\omega&\omega^2&1&1\\
0&1&0&1&\omega^2&\omega^2&\omega&0&\omega\\
0&0&1&\omega&\omega&0&\omega^2&\omega^2&1
                     \end{array}
 \right),
 G_{[12,3,8]}=\left( \begin{array}{ccccccccccccc}
       1&0&0&\omega&\omega^2&\omega&\omega&\omega&\omega&\omega^2&\omega&0\\
0&1&0&1&0&\omega&\omega^2&\omega^2&\omega^2&0&\omega&\omega\\
0&0&1&1&\omega&0&1&\omega&\omega^2&\omega^2&\omega&1
                     \end{array}
 \right),\\
  G_{[13,3,9]}=\left( \begin{array}{cccccccccccccc}
       1&0&0&1&\omega&1&\omega&\omega^2&0&\omega&\omega^2&1&1\\
0&1&0&0&0&1&1&1&1&1&1&\omega&1\\
0&0&1&\omega&1&0&0&1&1&\omega&\omega&1&\omega^2
                     \end{array}
 \right),
 G_{[14,3,10]}=\left( \begin{array}{ccccccccccccccc}
      1&0&0&0&1&1&\omega&0&1&\omega&1&1&1&1\\
0&1&0&\omega&\omega&1&\omega^2&1&0&1&1&\omega&1&\omega^2\\
0&0&1&1&1&1&1&1&\omega&1&0&0&\omega&\omega
                     \end{array}
 \right),\\
  G_{[17,3,12]}=\left( \begin{array}{cccccccccccccccccc}
     1&0&0&0&1&\omega^2&\omega^2&\omega^2&\omega^2&\omega^2&\omega&\omega^2&1&1&1&0&\omega\\
0&1&0&\omega^2&\omega^2&\omega^2&\omega^2&1&\omega&0&\omega^2&0&0&\omega&\omega&\omega&1\\
0&0&1&1&\omega&\omega&1&\omega&0&\omega^2&\omega&\omega&\omega&0&\omega&1&1
                     \end{array}
 \right),\\
 G_{[18,3,13]}=\left( \begin{array}{ccccccccccccccccccc}
    1&0&0&1&1&\omega&1&\omega&0&\omega^2&\omega&\omega&\omega&0&\omega&1&1&0\\
0&1&0&1&1&\omega^2&\omega&1&\omega^2&0&0&1&\omega^2&1&\omega^2&1&\omega^2&\omega^2\\
0&0&1&1&\omega&1&\omega&\omega^2&\omega&\omega^2&1&\omega&\omega&\omega&0&\omega^2&0&\omega^2
                     \end{array}
 \right),\\
 G_{[19,3,14]}=\left( \begin{array}{cccccccccccccccccccc}
   1&0&0&1&1&\omega&\omega&1&\omega&0&\omega^2&\omega&\omega&\omega&0&\omega&1&1&0\\
0&1&0&1&1&\omega^2&1&\omega&1&\omega^2&0&0&1&\omega^2&1&\omega^2&1&\omega^2&\omega^2\\
0&0&1&1&\omega&1&1&\omega&\omega^2&\omega&\omega^2&1&\omega&\omega&\omega&0&\omega^2&0&\omega^2
                     \end{array}
 \right),\\
 G_{[22,3,16]}=\left( \begin{array}{ccccccccccccccccccccccc}
   1&0&0&\omega&\omega&1&1&1&\omega&\omega^2&1&0&\omega^2
   &\omega&\omega&\omega^2&0&1&\omega^2&\omega&1&0\\
0&1&0&1&1&\omega&\omega&1&\omega&\omega^2&0&1&\omega&
\omega^2&0&1&\omega&\omega^2&0&1&\omega&\omega^2\\
0&0&1&0&0&0&0&1&1&1&1&1&1&1&1&1&1&1&1&1&1&1
                     \end{array}
 \right),\\
 G_{[23,3,16]}=\left( \begin{array}{cccccccccccccccccccccccc}
  1&0&0&0&0&0&0&1&1&1&1&1&1&1&1&1&1&1&1&1&1&1&1\\
0&1&0&1&1&1&1&0&0&0&0&1&1&1&1&\omega&\omega&\omega&\omega&
\omega^2&\omega^2&\omega^2&\omega^2\\
0&0&1&0&1&\omega&\omega^2&0&1&\omega&\omega^2&0&1&\omega&\omega^2&0&
1&\omega&\omega^2&0&1&\omega&\omega^2
                     \end{array}
 \right),\\
G_{[24,3,17]}=\left( \begin{array}{ccccccccccccccccccccccccc}
1&0&0&\omega^2&\omega^2&\omega^2&\omega&\omega&\omega&\omega&\omega^2&\omega^2&
 \omega^2&0&0&1&\omega&\omega&\omega&1&1&0&0&0\\
0&1&0&\omega&0&1&1&\omega&\omega&1&0&\omega^2&\omega&1&\omega^2&1&\omega&0&1&
\omega&0&1&\omega^2&\omega\\
0&0&1&0&\omega&\omega^2&\omega&1&\omega&1&1&\omega&\omega^2&\omega^2&1&\omega&
\omega^2&1&0&0&\omega&\omega&0&1
        \end{array}
 \right).
\end{array}
$$
\end{prop}

\begin{prop}\label{prop-2}
There are quaternary linear $[5,3,2]$, $[6,3,3]$, $[10,3,6]$, $[11,3,7]$, $[15,3,10]$, $[16,3,11]$, $[20,3,14]$, $[21,3,15]$ codes with one-dimensional Hermitian hull, which have the following generator matrix, respectively.
$$\begin{array}{l}
    G_{[5,3,2]}=\left(
\begin{array}{ccccc}
1&0&0&\omega&0\\
0&1&0&\omega&\omega^2\\
0&0&1&0&\omega
\end{array}
\right),
G_{[6,3,3]}=\left(
\begin{array}{cccccc}
  1&0&0&1&1&1\\
0&1&0&\omega^2&\omega^2&1\\
0&0&1&1&\omega&1
\end{array}
\right),
G_{[10,3,6]}=\left(
\begin{array}{cccccccccc}
  1&0&0&0&\omega^2&0&\omega&\omega&\omega&1\\
0&1&0&0&1&1&1&\omega^2&1&1\\
0&0&1&0&1&\omega^2&\omega&\omega&0&\omega
\end{array}
\right),\\
 G_{[11,3,7]}=\left( \begin{array}{ccccccccccc}
                       1&0&0&1&1&1&\omega&1&\omega&\omega&0\\
0&1&0&0&1&1&1&1&1&1&1\\
0&0&1&1&0&1&1&\omega&\omega&\omega^2&\omega^2
                     \end{array}
 \right),
 G_{[15,3,10]}=\left(\begin{array}{ccccccccccccccc}
                      1&0&0&0&1&\omega&0&1&\omega&\omega^2&\omega^2&\omega&1&0&1\\
0&1&0&0&0&\omega^2&\omega&1&\omega^2&1&0&1&\omega^2&1&\omega^2\\
0&0&1&0&\omega&1&\omega^2&1&0&1&\omega^2&1&\omega&\omega^2&1
                     \end{array}
 \right),\\
 G_{[16,3,11]}=\left(
 \begin{array}{cccccccccccccccc}
  1&0&0&1&1&\omega&\omega&1&\omega&0&\omega^2&\omega&\omega&\omega&\omega&1\\
0&1&0&1&1&\omega^2&1&\omega&1&\omega^2&0&0&1&\omega^2&\omega^2&\omega^2\\
0&0&1&1&\omega&1&1&\omega&\omega^2&\omega&\omega^2&1&\omega&\omega&0&0
 \end{array}
 \right),\\
 G_{[20,3,14]}=\left(
 \begin{array}{cccccccccccccccccccc}
  1&0&0&\omega&\omega^2&\omega&0&1&0&\omega^2&\omega&1&0&\omega&\omega^2&0&1&1&1&1\\
0&1&0&0&1&1&1&0&1&\omega&\omega^2&1&0&\omega&\omega^2&1&\omega^2&\omega&\omega&\omega^2\\
0&0&1&\omega^2&\omega&\omega^2&1&0&1&\omega&\omega^2&0&1&\omega^2&\omega&0&1&1&1&1
 \end{array}
 \right),\\
 G_{[21,3,15]}=\left(
 \begin{array}{ccccccccccccccccccccc}
   1&0&0&1&1&\omega&\omega&1&\omega&0&\omega^2&\omega&\omega&\omega&0&\omega&1&1&0&0&1\\
0&1&0&1&1&\omega^2&1&\omega&1&\omega^2&0&0&1&\omega^2&1&\omega^2&1&\omega^2&\omega^2&0&0\\
0&0&1&1&\omega&1&1&\omega&\omega^2&\omega&\omega^2&1&\omega&\omega&\omega&0&\omega^2&0&\omega^2&1&\omega
 \end{array}
 \right).
  \end{array}
$$
\end{prop}

Based on the above discussion, we can obtain the following table.

\begin{center}
\begin{tabular}{ccc|ccc}
\multicolumn{6}{c}{{\rm Table 3: Optimal quaternary $[n,3]$ codes with one-dimensional Hermitian hull}}\\
\hline
   $n$ & $D_4^H(n,3,1)$& $D_4(n,3)$~(Reference) & $n$&$D_4^H(n,3,1)$ & $D_4(n,3)$~(Reference)\\
    \hline\hline
   $4$ &$2$& 2~\cite[Table 4]{DM} & $22$ &$16$& 16\cite[Table 4]{DM}\\

   $5$ &$2$& 3\cite[Table 4]{DM} & $23$ &$16$& 16\cite[Table 4]{DM}\\

   $6$ &$3$& 4\cite[Table 4]{DM} & $24$ &$17$& 17\cite[Table 4]{DM}\\

   $7$ &$4$& 4\cite[Table 4]{DM} & $25$ &$18$& 18\cite[Table 4]{DM}\\

   $8$ &$5$& 5\cite[Table 4]{DM}& $26$ &$18$& 19\cite[Table 4]{DM}\\

   $9$ &$6$& 6\cite[Table 4]{DM} & $27$ &$19$& 20\cite[Table 4]{DM}\\

   $10$ &$6$& 6\cite[Table 4]{DM} & $28$ &$20$& 20\cite[Table 4]{DM}\\

   $11$ &$7$& 7\cite[Table 4]{DM} & $29$ &$21$& 21\cite[Table 4]{DM}\\

   $12$ &$8$& 8\cite[Table 4]{DM} & $30$ &$22$& 22\cite[Table 4]{DM}\\

   $13$ &$9$& 9\cite[Table 4]{DM} & $31$ &$22$& 23\cite[Table 4]{DM}\\

   $14$ &$10$& 10\cite[Table 4]{DM} & $32$ &$23$& 24\cite[Table 4]{DM}\\

   $15$ &$10$& 11\cite[Table 4]{DM} & $33$ &$24$& 24\cite[Table 4]{DM}\\

   $16$ &$11$& 12\cite[Table 4]{DM} & $34$ &$25$& 25\cite[Table 4]{DM}\\

   $17$ &$12$& 12\cite[Table 4]{DM} & $35$ &$26$& 26\cite[Table 4]{DM}\\

   $18$ &$13$& 13\cite[Table 4]{DM} & $36$ &$26$& 27 \cite[Table 3]{IT-1}\\

   $19$ &$14$& 14\cite[Table 4]{DM} & $48$ &$35$& 36 \cite[Table 4]{IT-1}\\

   $20$ &$14$& 15\cite[Table 4]{DM} & $52$ &$38$& 39 \cite[Table 4]{IT-1}\\

   $21$ &$15$& 16\cite[Table 4]{DM}&\\
    \hline
\end{tabular}
\end{center}

\begin{remark}
To save the space, the codes in Table 3 can be obtained from one of the authors' website, namely, {\tt https://ahu-coding.github.io/code2/}.
\end{remark}

\setlength{\arraycolsep}{3pt}

\begin{prop}\label{prop-a}
Suppose that $n=21s+n_0\geq 4$. Then we have
$$
\begin{array}{ll}
  D^H_4(21s+1,3,1)=16s, & D_4^H(21s+2,3,1)=16s  \vspace{1ex},\\
  D_4^H(21s+3,3,1)=16s+1, & D_4^H(21s+4,3,1)=16s+2 \vspace{1ex},\\
  D_4^H(21s+7,3,1)=16s+4, & D_4^H(21s+8,3,1)=16s+5 \vspace{1ex},\\
  D_4^H(21s+9,3,1)=16s+6, & D_4^H(21s+12,3,1)=16s+8 \vspace{1ex},\\
  D_4^H(21s+13,3,1)=16s+9, & D_4^H(21s+14,3,1)=16s+10 \vspace{1ex},\\
  D_4^H(21s+17,3,1)=16s+12, & D_4^H(21s+18,3,1)=16s+13 \vspace{1ex},\\
  D_4^H(21s+19,3,1)=16s+14. &
\end{array}
$$
\end{prop}

\begin{proof}
Combining Lemma \ref{lem-1}, Proposition \ref{prop-1} and Table 2, the result holds.
\end{proof}

\begin{prop}\label{prop-b}
Suppose that $n=21s+n_0\geq 6$. Then we have
$$
\begin{array}{ll}
  D^H_4(21s+6,3,1)=16s+3, & D_4^H(21s+10,3,1)=16s+6  \vspace{1ex},\\
  D^H_4(21s+11,3,1)=16s+7, & D_4^H(21s+15,3,1)=16s+10  \vspace{1ex},\\
  D^H_4(21s+16,3,1)=16s+11, & D_4^H(21s+20,3,1)=16s+14  \vspace{1ex},\\
  D^H_4(21s+21,3,1)=16s+15. &
\end{array}
$$
\end{prop}

\begin{proof}
Suppose that $s$ is a nonnegative integer.
By Table 3, we know that
$$
\begin{array}{lll}
  D^H_4(6,3,1)=3, & D_4^H(10,3,1)=6,&D^H_4(11,3,1)=7 \vspace{1ex},\\
  D_4^H(15,3,1)=10,& D^H_4(16,3,1)=11,& D_4^H(20,3,1)=14  \vspace{1ex},\\
   D_4^H(21,3,1)=15,&  D^H_4(27,3,1)=19, &D^H_4(31,3,1)=22  \vspace{1ex}.\\
\end{array}
$$
Applying Lemma \ref{lem-1} to the above codes, we obtain
$$
\begin{array}{ll}
  D^H_4(21s+6,3,1)\geq16s+3, & D_4^H(21s+10,3,1)\geq16s+6  \vspace{1ex},\\
  D^H_4(21s+11,3,1)\geq16s+7, & D_4^H(21s+15,3,1)\geq16s+10  \vspace{1ex},\\
  D^H_4(21s+16,3,1)\geq16s+11, & D_4^H(21s+20,3,1)\geq16s+14  \vspace{1ex},\\
  D^H_4(21s+21,3,1)\geq16s+15. &
\end{array}
$$

It follows from Table 3 that there is no quaternary $[n_0,3,d_0]$ linear code with one-dimensional Hermitian hull for
 $$(n_0,d_0)\in \{(16,12),(20,15),(21,16),(32,24),(36,27),(48,36),(52,39)\}.$$
 By Theorems \ref{thm-no} and \ref{thm-2}, there is no quaternary $[21s+n_0,3,16s+d_0]$ linear code $C$ with one-dimensional Hermitian hull and $d(C^{\perp_H})\geq 2$. Now suppose that there is a quaternary $[21s+n_0,3,16s+d_0]$ linear code $C_0$ with one-dimensional Hermitian hull and $d(C_0^{\perp_H})=1$. By deleting the zero column, we can obtain a quaternary $[21s+n_0-1,3,16s+d_0]$ linear code $C'$ with one-dimensional Hermitian hull. This contradicts the Griesmer bound (see Table 2).
Hence
$$
\begin{array}{ll}
  D^H_4(21s+6,3,1)\leq16s+3, & D_4^H(21s+10,3,1)\leq16s+6  \vspace{1ex},\\
  D^H_4(21s+11,3,1)\leq16s+7, & D_4^H(21s+15,3,1)\leq16s+10  \vspace{1ex},\\
  D^H_4(21s+16,3,1)\leq16s+11, & D_4^H(21s+20,3,1)\leq16s+14  \vspace{1ex},\\
  D^H_4(21s+21,3,1)\leq16s+15. &
\end{array}
$$
Therefore, we obtain the desired result.
\end{proof}

\begin{remark}
Combining Propositions \ref{prop-a} and \ref{prop-b}, we can determine the exact value of $D_4^H(n,3,1)$ except that $n\equiv~5~({\rm mod}~21)$, as shown in Table 4. Therefore, we also complete the proof of Theorem \ref{thm-1}. Only the case $n\equiv~5~({\rm mod}~21)$ is not completed, since the number of quaternary linear $[68, 3, 51]$ codes is too large.
\end{remark}

\begin{center}
\begin{tabular}{cc|cc|cc}
\multicolumn{6}{c}{{\rm Table 4: The characterization for $d_{one}^H(n,3),n\geq 4$}}\\
\hline
   $n$ & $D_4^H(n,3,1)$& $n$ & $D_4^H(n,3,1)$&$n$ & $D_4^H(n,3,1)$\\
    \hline\hline
   $21s$ &$16s-1$& $21s+7$ &$16s+4$ &$21s+14$ &$16s+10$ \\
   $21s+1$ &$16s$& $21s+8$ &$16s+5$&$21s+15$ &$16s+10$\\
   $21s+2$ &$16s$& $21s+9$ &$16s+6$&$21s+16$ &$16s+11$ \\
 $21s+3$ &$16s+1$& $21s+10$ &$16s+6$&$21s+17$ &$16s+12$ \\
 $21s+4$ &$16s+2$& $21s+11$ &$16s+7$&$21s+18$ &$16s+13$ \\
 $21s+5$ &$\geq 16s+2$& $21s+12$ &$16s+8$&$21s+19$ &$16s+14$ \\
 $21s+6$ &$16s+3$& $21s+13$ &$16s+9$&$21s+20$ &$16s+14$ \\
    \hline
\end{tabular}
\end{center}

\subsection{Optimal quaternary linear $[n,n-2]$ codes with one-dimensional Hermitian hull}
In this section, we determine the exact value of $D_4^K(n,n-2,1)$ for $n\geq 3$.

\begin{theorem}
Suppose that $n>2$ is an integer. Then we have
$$D_4^H(n,n-2,1)= \left\{
\begin{array}{ll}
 3,~& {\rm if}~ n=4,\\
 2,~&{\rm if}~n=3~{\rm or}~n\geq 5.
\end{array}
\right.$$
\end{theorem}

\begin{proof}
By Theorem \ref{th-k=1}, $D_4^H(3,1,1)=2$. By Theorem \ref{conjecture}, $D^H_4(4,2,1)=3$. By Table 3, $D^H_4(5,3,1)=2$.
If $D_4(n,n-2)\geq 3$, then it follows from the sphere-packing bound that
 $$4^{n-2}\leq \frac{4^n}{1+3n}.$$
If $n>5$, then $4^{n-2}> \frac{4^n}{1+3n},$ which contradicts the sphere-packing bound.
Hence $$D^H_4(n,n-2,1)\leq D_4(n,n-2)\leq 2.$$

\begin{itemize}
  \item Suppose that $n$ is odd. It follows from Theorem \ref{th-k=1} that
$$D^H_4(n,n-2,1)\geq D^H_4(n-1,n-2,1)=2.$$
Hence $D_4^H(n,n-2,1)=2$.
  \item
Suppose that $n$ is even. Consider the code $C$ with the following parity-check matrix
$$H=\left[
       \begin{array}{cccccccccc}
         1 & 0 & 1& 1&1&\ldots  & 1  \\
         0 & 1 & \omega&  \omega^2&0&\ldots & 0
       \end{array}
     \right]_{2\times n}.
$$

Then it can be checked that $C$ has parameters $[n,n-2,2]$ and $C$ has one-dimensional Hermitian hull. Hence $D_4^H(n,n-k,1)=2$.
\end{itemize}
This completes the proof.
\end{proof}

\subsection{Optimal quaternary linear $[n,n-3]$ codes with one-dimensional Hermitian hull}

In this section, we determine the exact value of $D_4^H(n,n-3,1)$ for $n\geq 4$.

\begin{lem}\label{lemma-d(n-k)}
Let $k\geq 3$ and $n\geq \frac{4^{k}-1}{3}-1$. Then $D^H_4(n,n-k,1)=2$.
\end{lem}

\begin{proof}
Suppose that $n= \frac{4^{k}-1}{3}$. Assume that $C$ is a quaternary linear $[n,n-3,3]$ code. Then $C^{\perp_H}$ is a quaternary linear $[n,3]$ code with $d((C^{\perp_H})^{\perp_H})=d(C)=3$. Then there exists a vector
${\bf m}=(m_1,\ldots,m_{\frac{4^k-1}{3}})$ such that $C^{\perp_H}$ is equivalent to the code $C_k({\bf m})$ with the generator matrix $G_k({\bf m})$. Since $C$ has the minimum distance $3$, $G_k({\bf m})$ does not have the same two columns, i.e., $m_i\leq 1$ for $1\leq i\leq \frac{4^k-1}{3}$.
By $\sum_{i}^{\frac{4^k-1}{3}}m_i=n=\frac{4^k-1}{3}$, we have
$m_i=1$. Hence ${\bf m}=(1,\ldots,1)$ and $C_k({\bf m})$ is a quaternary linear $[\frac{4^k-1}{3},k,4^{k-1}]$ code.
By Theorem \ref{thm-2}, $C_k({\bf m})$ is Hermitian SO. It turns out that
$$\dim({\rm Hull_H}(C))=\dim({\rm Hull_H}(C^{\perp_H}))=\dim({\rm Hull_H}(C_k({\bf m})))=k\geq 3.$$
This implies that $D_4^H(\frac{4^{k}-1}{3},n-k,1)\leq 2$.

Suppose that $n= \frac{4^{k}-1}{3}-1$. Assume that $C$ is a quaternary linear $[n,n-3,3]$ code. Then $C^{\perp_H}$ is a quaternary linear $[n,3]$ code with $d((C^{\perp_H})^{\perp_H})=d(C)=3$. Then there exists a vector
${\bf m'}=(m'_1,\ldots,m'_{\frac{4^k-1}{3}})$ such that $C^{\perp_H}$ is equivalent to the code $C_k({\bf m'})$ with the generator matrix $G_k({\bf m'})$. Since $C$ has the minimum distance $3$, $G_k({\bf m'})$ does not have the same two columns, i.e., $m'_i\leq 1$ for $1\leq i\leq \frac{4^k-1}{3}$.
By $\sum_{i}^{\frac{4^k-1}{3}}m'_i=n=\frac{4^k-1}{3}-1$, we have $c_{i_0}=0$ for some $i_0$ where $1\leq i_0\leq \frac{4^k-1}{3}$. Hence $C_k({\bf m'})$ is the punctured code of $C_k({\bf m})$ on the $i_0$-th coordinate.
By Theorem \ref{thm-luo},
$$\dim({\rm Hull_H}(C))=\dim({\rm Hull_H}(C^{\perp_H}))=\dim({\rm Hull_H}(C_k({\bf m})))-1=k-1\geq 2.$$
This implies that $D_4^H(\frac{4^{k}-1}{3}-1,n-k,1)\leq 2$.

Suppose that $n>\frac{4^{k}-1}{3}$. If $D_4(n,n-k)\geq 3$, then
$4^{n-k}> \frac{4^n}{1+n},$
which contradicts the sphere-packing bound.
Hence $D^H_4(n,n-k,1)\leq D_4(n,n-k)\leq 2$.
Consider the code $C$ with the following generator matrix
$$G=\left[
       \begin{array}{cccccccccc}
         1 & 0 & \ldots & 0 & 1 & 1 & 1 &0&\ldots&0 \\
         0 & 1 & \ldots & 0 & 1 & 1 & 0 &0&\ldots&0\\
         \vdots & \vdots & \ddots & \vdots & \vdots & \vdots & \vdots& \vdots & \ddots & \vdots  \\
         0 & 0 & \cdots & 1 & 1 & 1 & 0&0&\ldots&0
       \end{array}
     \right]_{(n-k)\times n}.
$$
Then it can be checked that $C$ has parameters $[n,n-k,2]$ and one-dimensional Hermitian hull. Hence $D_4^H(n,n-k,1)=2$. This completes the proof.
\end{proof}

\begin{theorem}\label{thm-n-3}
Suppose that $n>3$ is an integer. Then we have
$$D_4^H(n,n-3,1)= \left\{
\begin{array}{ll}
 4,~ &{\rm if}~ n=4,\\
 3,~& {\rm if}~ 5\leq n\leq 19,\\
 2,~&{\rm if}~n\geq 20.
\end{array}
\right.$$
\end{theorem}

\begin{proof}
By Theorem \ref{th-k=1}, we know that $D_4^H(4,1,1)=4$. By Lemma \ref{lemma-d(n-k)}, we know that $D_4^H(n,n-3,1)=2$ for $n\geq 20$. By Sections 4.2 and 4.3, we know that $D_4^H(5,2,1)=3$ and $D^H_4(6,3,1)=3$.

By BKLC database of Magma \cite{magma}, there exists a quaternary linear $[21,18,3]$ code $C^3_{[21,18,3]}$ with 3-dimensional Hermitian hull. Shortening the code $C^3_{[21,18,3]}$ on the sets $\{11,17\}$, $\{2,6,12\}$, $\{4,7,10,12\}$, $\{8,10,11,16,17\}$, $\{6,11,13,15,16,18\}$,
$\{1,10,11,14,16,\\20,21\}$, $\{1,2,5,6,9,10,13,14\}$ and $\{1,2,3,4,5,6,12,15,18\}$, one can construct quaternary linear $[19,16,3]$, $[18,15,3]$, $[17,14,3]$, $[16,13,3]$, $[15,12,3]$, $[14,11,3]$, $[13,10,3]$ and $[12,9,3]$ codes with one-dimensional Hermitian hull.

By BKLC database of Magma \cite{magma}, there exists a quaternary linear $[11,8,3]$ code $C^1_{[11,8,3]}$ with one-dimensional Hermitian hull. Shortening the code $C^1_{[11,8,3]}$ on the sets $\{3\}$, $\{2,11\}$, $\{1,2,4\}$, $\{2,7,9,11\}$ and $\{1,5,6,9,10\}$. one can construct quaternary linear $[10,7,3]$, $[9,6,3]$, $[8,5,3]$, $[7,4,3]$ and $[6,3,3]$ codes with one-dimensional Hermitian hull.
\end{proof}

\begin{remark}
To save the space, the codes in Theorem \ref{thm-n-3} can be obtained from one of the authors' website, namely, {\tt https://ahu-coding.github.io/code2/}.
\end{remark}

\section{Quaternary linear codes with one-dimensional Hermitian hull of lengths up to 12}

In this section, we determine the largest minimum weight $D_4^H(n,k,1)$ among all quaternary linear $[n, k]$ codes with one-dimensional Hermitian hull for $1\leq k<n\leq 12$. In Section 4, we determine the exact values of $D_4^H(n,k,1)$ and $D_4^H(n,n-k,1)$ for $k\in \{1,2,3\}$. Hence we only consider $4\leq k\leq n-4\leq 8$.
Based on the classification in \cite{D-1}, we find $D_4^H(n,k,1)<D_4(n,k)$ for $(n,k)\in \{(10,4),(11,5),(11,6)\}.$

\begin{itemize}
  \item [(1)] By BKLC database of Magma \cite{magma}, there exists a quaternary LCD $[13,9,4]$ code $C^0_{[13,9,4]}$. Shortening the code $C^0_{[13,9,4]}$ on coordinates sets $\{4\}$, $\{1,8\}$, $\{1,4,11\}$, $\{4,7,10,12\}$ and $\{5,7,10,11,13\}$, one can obtain quaternary linear $[12,8,4]$, $[11,7,4]$, $[10,6,4]$, $[9,5,4]$ and $[8,4,4]$ codes with one-dimensional Hermitian hull.
  \item [(2)] By BKLC database of Magma \cite{magma}, there exist quaternary linear $[10,5,5]$ and $[12,6,6]$ codes with one-dimensional Hermitian hull. We can construct a quaternary linear $[9,4,5]$ code with one-dimensional Hermitian hull and its generator matrix is
   $$G_{[9,4,5]}=\left(\begin{array}{ccccccccc}
             1&0&0&0&\omega^2&1&0&1&\omega^2\\
0&1&0&0&\omega^2&1&1&\omega^2&0\\
0&0&1&0&0&1&1&1&1\\
0&0&0&1&1&0&1&\omega^2&\omega^2
           \end{array}
   \right).$$
  \item [(3)] By BKLC database of Magma \cite{magma}, there exists a quaternary LCD $[13,6,6]$ code $C^0_{[13,6,6]}$. Shortening the code $C^0_{[13,6,6]}$ on the set $\{4\}$, one can construct a quaternary linear $[12,5,6]$ code with one-dimensional Hermitian hull.
  \item [(4)] By BKLC database of Magma \cite{magma}, there exists a quaternary linear $[14,4,9]$ code $C^1_{[14,4,9]}$ with one-dimensional Hermitian hull. Puncturing  the code $C^1_{[14,4,9]}$ on the sets $\{9,13\}$ and $\{1,10,12\}$, one can construct quaternary linear $[12,4,7]$ and $[11,4,6]$ codes with one-dimensional Hermitian hull.
  \item [(5)] By adding the zero column, we have $D^H_4(10,4,1)\geq D^H_4(9,4,1)=5$, $D^H_4(11,5,1)\geq D^H_4(10,5,1)=5$, $D^H_4(11,6,1)\geq D^H_4(10,6,1)=4$, $D^H_4(12,7,1)\geq D^H_4(11,7,1)=4$.
\end{itemize}
Combined with the above discussion, we give Table 5.

\begin{table}[h]\setlength\tabcolsep{12pt}
\centering
\begin{threeparttable}
\begin{tabular}{cccccccccccc}
\multicolumn{12}{c}{{\rm Table 5: The exact value of $D^H_4(n,k,1)$ for $n\leq 12$}}\\
\hline
   $n\backslash k$ &1&2 & $3$ & $4$ &5&6&7&8&9&10&11\\
    \hline\hline
    2&2&&&&&&&&&&\\
    3&2&1&&&&&&&&&\\
   $4$ &$4$ &3 &$2$ & &&&&&&&\\
   $5$ & $4$& 3 &$2$ &1&& &&&&&\\
   $6$ & $6$&4&$3$ & 2 &2&&&&&&\\
   $7$ & $6$&5  &$4$ &3 &2&1&&&&&\\
   $8$ &$8$ &5 &$5$ &4  &3&2&2&&&&\\
   $9$ &$8$ &7 &$6$ &5 &4&3&2&1&&&\\
   $10$ &$10$&7 &$6$ &5  &5 &4&3&2&2&&\\
   $11$ &$10$&8 &$7$ & 6 &5 &4&4&3&2&1&\\
   $12$ &$12$&9 &$8$ &7  &6 &6&4 &4&3&2&2\\
    \hline
\end{tabular}
\begin{tablenotes}
\footnotesize
\item The upper bound can refer to \cite{codetables}.
\end{tablenotes}
\end{threeparttable}
\end{table}

\begin{remark}
The value in Table 5 denotes the minimum distance of an optimal
quaternary linear $[n, k]$ code with one-dimensional Hermitian hull. All computations in this paper have been done by MAGMA \cite{magma}. To save the space, the codes in Table 1 can be obtained from one of the authors' website, namely, {\tt https://ahu-coding.github.io/code2/}.
\end{remark}

\section{Applications to EAQECCs}
In this section, we introduce some definitions about EAQECCs and construct some EAQECCs from quaternary linear codes with one-dimensional Hermitian hull.
We use $[[n,k,d;c]]_q$ to denote a $q$-ary EAQECC that encodes $k$ information qubits into $n$ channel qubits with the help of $c$ pre-shared entanglement
pairs. EAQECCs were introduced by Brun et al. in \cite{Eaqecc}, which include the standard quantum stabilizer codes as a special case.

\begin{prop}{\rm\cite{DCC-hull}}\label{prop-EAqecc}
Let $C$ be a quaternary $[n,k,d]$ linear code and $C^{\perp_H}$ be its Hermitian dual code with parameters $[n,n-k,d^{\perp_H}]_4$, where $d^{\perp_H}$ is the minimum distance of $C^{\perp_H}$. Assume that $\dim({\rm Hull}_{\rm H}(C))=1$. Then, there exist an $[[n,k-1,d;n-k-1]]_2$ EAQECC and an $[[n,n-k-1,d^{\perp_H};k-1]]_2$ EAQECC.
\end{prop}

Combining Table 4 and Proposition \ref{prop-EAqecc}, we have the following corollary.

\begin{cor}
Suppose that $n=21s+t\geq 4$. There is an $[[n,k,d;c]]_2$ EAQECC for the following parameters.
$$\begin{array}{ll}
    [[21s,2,16s-1;21s-4]]_2, & [[21s+11,2,16s+7;21s+7]]_2, \\
    {[[21s+1,2,16s;21s-3]]}_2, & {[[21s+12,2,16s+8;21s+8]]}_2,\\
    {[[21s+2,2,16s;21s-2]]}_2, & {[[21s+13,2,16s+9;21s+9]]}_2,\\
    {[[21s+3,2,16s+1;21s-1]]}_2, & {[[21s+14,2,16s+10;21s+10]]}_2,\\
    {[[21s+4,2,16s+2;21s]]}_2, & {[[21s+15,2,16s+10;21s+11]]}_2,\\
    {[[21s+5,2,16s+2;21s+1]]}_2, & {[[21s+16,2,16s+11;21s+12]]}_2,\\
    {[[21s+6,2,16s+3;21s+2]]}_2, & {[[21s+17,2,16s+12;21s+13]]}_2,\\
    {[[21s+7,2,16s+4;21s+3]]}_2, & {[[21s+18,2,16s+13;21s+14]]}_2,\\
    {[[21s+8,2,16s+5;21s+4]]}_2, & {[[21s+19,2,16s+14;21s+15]]}_2,\\
    {[[21s+9,2,16s+6;21s+5]]}_2, & {[[21s+20,2,16s+14;21s+16]]}_2,\\
    {[[21s+10,2,16s+6;21s+6]]}_2.&
  \end{array}
$$
\end{cor}

\begin{example}
By Table 4, there is a quaternary $[12,3,8]$ linear code with one-dimensional Hermitian hull. By Proposition \ref{prop-EAqecc}, we obtain a binary EAQECC with parameters $[[12,2,8;8]]_2$, which has a better minimum distance and smaller amount of entanglement than the best known EAQECC $[[12,2,7;9]]_2$ (see \cite{codetables}).

By Table 5, there is a quaternary $[12,4,7]$ linear code with one-dimensional Hermitian hull. By Proposition \ref{prop-EAqecc}, we obtain a binary EAQECC with parameters $[[12,3,7;7]]_2$, which has smaller amount of entanglement than the best known EAQECC $[[12,3,7;9]]_2$ (see \cite{codetables}).

Combining Table 4, Table 5 and Proposition \ref{prop-EAqecc}, we give Tables 6 and 7. The parameters in bold denote that the corresponding code has different parameters according to \cite{codetables}.
\end{example}

\begin{center}
\begin{tabular}{cccccccccccc}
\multicolumn{12}{c}{{\rm Table 6: $[[n,k,d;c]]_2$ EAQECCs with $[d;c]$ for $n\leq 12$}}\\
\hline
   $n\backslash k$ &0&1&2 & $3$ & $4$ &5&6&7&8&9&10\\
    \hline\hline
    2&[2;0]&&&&&&&&&&\\
    3&[2;1]&[1;0]&&&&&&&&&\\
   $4$ &[4;2] &[3;1] &[2;0] & &&&&&&&\\
   $5$ & [4;3]& [3;2] &[2;1] &[1;0]&& &&&&&\\
   $6$ & [6;4]&[4;3]&[3;2] & [2;1] &[2;0]&&&&&&\\
   $7$ & [6;5]&[5;4]  &[4;3] &[3;2] &[2;1]&[1;0]&&&&&\\
   $8$ &[8;6] &[5;5] &[5;4] &{\bf[4;3]}  &[3;2]&[2;1]&[2;0]&&&&\\
   $9$ &{\bf[8;7]} &[7;6] &[6;5] &{\bf[5;4]} &[4;3]&[3;2]&[2;1]&[1;0]&&&\\
   $10$ &[10;8]&[7;7] &[6;6] &[5;5]  &[5;4] &{\bf[4;3]}&[3;3]&[2;1]&[2;0]&&\\
   $11$ &{\bf[10;9]}&[8;8] &{\bf[7;7]} & [6;6] &[5;5] &[4;4]&[4;3]&[3;2]&[2;1]&[1;0]&\\
   $12$ &[12;10]&[9;9] &{\bf[8;8]} &{\bf[7;7]}  &[6;6] &[6;5]&[4;4] &{\bf[4;3]}&[3;2]&[2;1]&[2;0]\\
    \hline
\end{tabular}
\end{center}

\begin{table}[h]\renewcommand{\arraystretch}{1.2}
\centering
\begin{tabular}{l|c|c}
\multicolumn{3}{c}{{\rm Table 7: $[[n,k,d;c]]_2$ EAQECCs for $k=2$}}\\
\hline
  The known EAQECCs  & Our parameters (Table 4) & The related EAQECCs \\
    \hline\hline
\makecell[l]{$[[13,2,4;0]]_2$ \cite{codetables}\\
    $[[13,2,5;1]]_2$ \cite{codetables}\\
$[[13,2,6;2]]_2$\cite{codetables}\\
$[[13,2,8;9]]_2$\cite{codetables}}&$[13,3,9]_4$ & ${\bf [[13,2,9;9]]_2}$ \\
\hline

  \makecell[l]{$[[14,2,5;0]]_2$ \cite{codetables}\\
    $[[14,2,7;2]]_2$ \cite{codetables}\\
$[[14,2,9;9]]_2$\cite{codetables}} &$[14,3,10]_4$ & ${\bf[[14,2,10;10]]_2}$ \\
    \hline

 \makecell[l]{$[[16,2,6;0]]_2$ \cite{codetables}\\
    $[[16,2,8;2]]_2$ \cite{codetables}\\
$[[16,2,10;9]]_2$\cite{codetables}} &$[16,3,11]_4$ & ${\bf[[16,2,11;12]]_2}$ \\
    \hline

    \makecell[l]{$[[17,2,6;0]]_2$ \cite{codetables}\\
    $[[17,2,8;2]]_2$ \cite{codetables}\\
$[[17,2,10;9]]_2$\cite{codetables}} &$[17,3,12]_4$ & ${\bf[[17,2,12;13]]_2}$ \\
    \hline

     \makecell[l]{$[[18,2,6;0]]_2$ \cite{codetables}\\
    $[[18,2,8;2]]_2$ \cite{codetables}\\
$[[18,2,10;9]]_2$\cite{codetables}} &$[18,3,13]_4$ & ${\bf[[18,2,13;14]]_2}$ \\
    \hline

      \makecell[l]{$[[19,2,6;0]]_2$ \cite{codetables}\\
    $[[19,2,9;2]]_2$ \cite{codetables}\\
$[[19,2,10;9]]_2$\cite{codetables}} &$[19,3,14]_4$ & ${\bf[[19,2,14;15]]_2}$ \\
    \hline

     \makecell[l]{$[[20,2,6;0]]_2$ \cite{codetables}\\
$[[20,2,10;2]]_2$\cite{codetables}} &$[20,3,14]_4$ & ${\bf[[20,2,14;16]]_2}$ \\
    \hline

    \makecell[l]{$[[22,2,6;0]]_2$ \cite{codetables}\\
    $[[22,2,7;1]]_2$ \cite{codetables}\\
$[[22,2,11;2]]_2$\cite{codetables}} &$[22,3,16]_4$ & ${\bf[[22,2,16;18]]_2}$ \\
    \hline
\end{tabular}
\end{table}

\section{Conclusion}
We have studied some properties of quaternary linear codes with one-dimensional Hermitian hull. We have determined the exact values of $D_4^H(n,k,1)$ and $D_4^H(n,n-k,1)$ for $n\leq 12$ or
$k\in\{1,2,3\}$, which solved a conjecture proposed by Mankean and Jitman \cite{AJM-conjecture}. Finally, as
an application, we constructed some EAQECCs with new parameters from quaternary linear codes with one-dimensional Hermitian hull.

As future work, it would be an interesting problem to consider the exact value of $D_4^H(n,k,l)$ for $l\geq 2$ or to extend the codetable to longer lengths.
\section*{Conflict of Interest}
The authors have no conflicts of interest to declare that are relevant to the content of this article.

\section*{Data Deposition Information}
Our data can be obtained from the authors upon reasonable request.

\section*{Acknowledgement}
This research is supported by Natural Science Foundation of China (12071001).

\end{document}